\documentclass[twoside,10pt]{article}
\usepackage[utf8]{inputenc}
\usepackage{t1enc}
\usepackage{amsmath}
\usepackage{amssymb}
\usepackage{ntheorem-hyper}
\usepackage{enumerate}
\newcommand{\Vx}{\underline{x}}

\newcommand{\Vy}{\underline{y}}
\newcommand{\VX}{\underline{X}}

\newcommand{\irott}{\mathcal}
\DeclareMathOperator{\Tr}{Tr}
\theoremstyle{plain}
\newtheorem{Def}{Definition}
\newtheorem{remark}{Remark}
\newtheorem{Tet}{Theorem}
\newtheorem{Lem}{Lemma}
\newenvironment{proof}[1][Proof]{\textbf{#1:}}{\hfill$\blacksquare$\newline}

\renewcommand{\S}{S}
\def\iM{{\cal M}}

\def\iE{{\cal E}}

\def\eps{\varepsilon}
\def\Prob{\mbox{Prob}\,}
\def\rrho{\omega}
\title{Reaching the Holevo Capacity via von Neumann measurement, and its use}
\author{Farkas Lóránt}

\begin{document}
\maketitle
\begin{section}{Introduction}
The method of types plays very important and central role in the classical information theory. With it the central theorems can be easily and fast proved. This work try to generalise the main ideas of the method of types, and prove one of the central theorems - Reaching the Holevo capacity - in quantum environment. 

The main problem is: Suppose that we want to send information with quantum's. This problem is relevant, because 1) quantum computers would prefer this way of communication 2) The miniaturisation in the Information Technologies can lead to these type of problem. The problem can be formalised as follows: We code our classical message to quantum sequences states from $(\omega_1,\omega_2,\dots,\omega_l)$. We suppose that there is a unique non-reversible quantum transformation $\iE(\cdot)$ - quantum channel - which acts on every of these quantum's. The question is how many bits of information can be transmitted by one quantum. A theorem stated by Gordon and Levitin, proved by Holevo \cite{Holevo}, gives an upper bound to the amount of information that can be communicated. If the sender codes his information to quantum states with density matrix $\rho_i$ with a priori probabilities $p_l$ then the communicated information cannot be bigger than
\begin{equation}\S(\sum_{i=1}^l p_i \rho_i)-\sum_{i=1}^l p_i S(\rho_i)\end{equation}
where $S(\cdot)$ is the von Neumann entropy. If the outcome of the channel is $\rho_l=\iE(\omega_l)$ this gives an upper bound. So the problem is to show, that this bound can be reached.

At the end, our result is stronger than the work of Holevo \cite{Holevo2} or Schumacher and Westmoreland \cite{Schumacher}, because we will show that the decoding can be done by von Neumann measurement, not only with POVM (We doesn't use ''Pretty Good Measurement`` as in \cite{P-Good-Mes} or \cite{Holevo2}). Moreover, we show two use of the von Neumann measurement. 

The first use is that the procedure can be generalised to finite compound channel, that means, we can create an optimal coding scheme to work not only with one quantum channel, but with finite many. Definition and capacity is in section \ref{Fehaszn1} 

The other use is that we can translate decoding of classical information to decoding of classical quantum information, and with the von Neumann measurement the time of the measuring procedure can be extremely shorten, which means that classical information can be decoded by a quantum apparatus in linear! time (This is a strong result, the best codes which reach the Shannon's bound needs $n^{\log_2(n)}$ time to decode).

This work base notation is borrowed from the work of Schumacher and Westmoreland \cite{Schumacher}, but the base ideas of the proofs, comes from the classical theory e.g. \cite{Csiszar}.
\end{section}

\section{Notations and basic lemmas}\label{Jel}
Let $\iE(\,\cdot\,)$ be a given quantum channel. Assume that $\rrho_1,\rrho_2,\dots,
\rrho_l=\rrho_1^l$ are input density matrices, with same dimension $d$ and $P=(p_1,p_2,\dots,p_n)$ is a 
probability distribution such that they  maximise the Holevo quantity
\begin{equation}
\chi(\iE,P,\rrho_1^l)=S(\iE(\rrho))-\sum_{i=1}^l p_i S(\iE(\rrho_i))\, , 
\end{equation}
where $\rrho=\sum_{i=1}^l p_i\rrho_i$.
Denote the possible outputs of the quantum channel by $\rho_i=\iE(\rrho_i),
\rho=\iE(\rrho)$, these are represented by $d\times d$ density matrices.

Fix $n$, the length of the (quantum) codewords. We generate randomly 
$M=2^{nR}$ piece codewords of length $n$ with probability distribution 
$P$. We denote these randomly generated sequence by $\alpha_i, \quad 1
\leq i \leq M$. If a statement is true all of the index $1\leq i \leq M$, 
then we will say that it is true for $\alpha$. The $j$-th symbol of 
$\alpha$ will be denoted by $\alpha(j)$. 

For all sequence we can define a quantum sequence as follows:
\begin{equation}
\rho_\alpha=
\rho_{\alpha(1)}\otimes\rho_{\alpha(2)}\otimes\dots\otimes\rho_{\alpha(n)}
\end{equation}

We will denote by $S(\rho|\alpha)$ the quantity $\sum_{j=1}^l 
p_j S(\rho_j)$, because $nS(\rho|\alpha)$ is the expected 
value of the Neumann entropy of the quantum sequence if we know which 
randomly generated sequence ($\alpha$) was sent (so we know which basis to use), while $S(\rho)$ is the von Neumann entropy of 
the sequence if we do not know which message was sent. 
So for fixed $P$ and $\rrho$ the Holevo capacity becomes
\begin{equation}
\chi = S(\rho)-\S(\rho|\alpha)
\end{equation}
Which resembles the Shannon capacity
\begin{equation}
C = H(Y)-H(Y|X)
\end{equation}
where $Y$ is the output random variable and $X$ is an input random random variable of the channel.

For a fixed $\varepsilon>0$, 
a sequence $\alpha$ is called $\eps$-typical with respect to $P$ if
\begin{equation}
2^{-n(H(P)+\varepsilon)} \leq P^n(\alpha)\leq 2^{-n(H(P)-\varepsilon)},
\end{equation}
where $H(P)$ is the Shannon entropy of $P$. We know from the law of 
large numbers that, if $n$ is  large enough then the probability
\begin{equation}
\Prob(\alpha \mbox{\ is typical})\geq 1-\varepsilon
\end{equation}
(because $\alpha$ was generated by distribution $P$), see \cite{Csiszar}.

Let the spectral decomposition of $\iE(\rho_\alpha)=\sum_{k=1}^{d^n}\lambda_{\alpha,k} |s_{\alpha,k}\rangle \langle s_{\alpha,k}|$. Because $\rho_{\alpha}$ is a tensor product, the eigenvectors are tensor products of eigenvectors of the $\rho_{\alpha(1)},\rho_{\alpha(2)},\dots, etc.$. So a measurement in the eigenbasis can be represented by a sequence, from numbers $\{1,2,\dots, d\}$ where the i-th term gives that what we would measure if we measure $\rho_{\alpha(i)}$ in its eigenbasis. Denote this correspondence by $s:\{1,2,\dots d^n\} \rightarrow \{1,2,\dots d\}^n$, note that $\lambda_{k}=\lambda_{s(k)_1}\lambda_{s(k)_2}\cdots \lambda_{s(k)_n}$. An eigenvector $|s_{\alpha,k}\rangle$ is $\delta$-typical if the above defined distribution ($\lambda_{s(k)_1}\lambda_{s(k)_2}\cdots$) is conditionally typical to the sequence $\alpha$ (see \cite{Csiszar})
\begin{equation}
-n(\S(\rho|\alpha)+\delta) \leq \log \lambda_{\alpha,k}
\leq -n(\S(\rho|\alpha)-\delta)
\end{equation}
Note that all exponent and logarithm are base of 2 across of this article.
The above definition means that if we define the typical projection as 
\begin{equation}
 Pi_\alpha=\sum_{k:s_{\alpha,k} \textnormal{is typical} } |s_{\alpha,k}\rangle \langle s_{\alpha,k}| \label{hiv2}
\end{equation}
, then 
\begin{equation}
dim(\Pi_\alpha)\leq 2^{n{\S(\rho|\alpha)+\delta}} \label{kifgasolt1}
\end{equation}
while 
\begin{align}
 1 & = \sum_{k=1}^{d^n}\lambda_{\alpha,k}\geq \sum_{k:s_{\alpha,k} \textnormal{is typical} } \lambda_{\alpha,k}\geq \sum_{k:s_{\alpha,k} \textnormal{is typical}}2^{-n(\S(\rho|\alpha)+\delta)}\\
2^{n(\S(\rho|\alpha)+\delta)}&\geq \sum_{k:s_{\alpha,k} \textnormal{is typical}}1=dim(\Pi_\alpha)
\end{align}

The pair $(\alpha,k)$ is $\delta$ typical if
\begin{equation}
-n\Big(H(P)+\S(\rho|\alpha)+\delta\Big)
\leq \log P^n(\alpha)\lambda_{\alpha,k}\leq -n\Big(H(P)+
\S(\rho|\alpha)-\delta\Big).
\end{equation}

Let the distribution of $(\alpha,s(k))$ be denoted by \begin{equation}
P_{\alpha,s(k)}=\Pi_{i=1}^n P(\alpha(i))\lambda_{\alpha(i),s(k)_i},
\end{equation}
it can be seen that this is a probability of independent, identically distributed random variables. The (Shannon) entropy of this distribution is 
\begin{equation}
H(P)+\sum_{i=1}^l p_i\S(\rho_i).
\end{equation}
So for these pair of random variables, the 
law of large numbers also true, so then by summing the probability 
of all typical pair we also get a greater number than $1-\varepsilon$ if $n$ is large enough. 
Suppose that the indexing of the eigenvalues is such that the typical eigenvalues are the first $d(\alpha)$ indices. Then
\begin{equation}
\sum_{\alpha}p_{\alpha} \sum_{i=1}^{d(\alpha)}\lambda_{\alpha,i} 
\geq 1-2\varepsilon \label{kifgasolt3}
\end{equation}
because from the sum we only left the atypical $\alpha$ (which 
probability is smaller that $\varepsilon$) and atypical $\alpha,k$ 
pairs (which probability is also smaller that $\varepsilon$)
 From this we know, that for $\tilde\rho=\sum_{\alpha}p_{\alpha} \sum_{i=1}^{d(\alpha)}\lambda_{\alpha,i} |s_{\alpha,i}\rangle \langle s_{\alpha,i}|$ it is true that $\Tr\tilde\rho\geq 1-2\varepsilon$.
 Define $\tilde\rho_\alpha$ as
\begin{equation}
 \tilde\rho_\alpha=\Pi_{\alpha} \rho_\alpha \Pi_{\alpha} \label{hiv0}
\end{equation}
 \begin{remark}\label{kifgasolt4}
	$\tilde\rho\leq \rho^{\otimes n}$ and $\tilde\rho_\alpha \leq \rho_\alpha$, and $\mathbb{E}(\tilde\rho_{\alpha})=\tilde\rho$
 \end{remark}

\begin{Lem}\label{eri=r}
For every $n\in \mathbb{N}$
\begin{equation}
\mathbb{E}[\rho_{\alpha}]=\rho^{\otimes n}
\end{equation}
\end{Lem}
\begin{proof}{Total induction on $n$}
For $n=1$ the equivalence is true by the definition of $\rho$.
Suppose that for $n=k-1$ the statement is true, then
For $n=k$ let denote $\alpha'$ an arbitrary $k-1$ length sequence then for every $\alpha$ $k$ length sequence can be written as $\alpha=(\alpha',i),\quad 1\leq i\leq l$. Then
\begin{align}
\mathbb{E}[\rho_{\alpha}]&=\sum_{\alpha}p^k(\alpha)\rho_\alpha=\sum_{\alpha'}\sum_{i=1}^l p^{k-1}(\alpha')p(i)\rho_{\alpha'}\otimes\rho_i=\\
&=\sum_{\alpha'}p^{k-1}(\alpha') \sum_{i=1}^l p_i\rho_{\alpha'}\otimes\rho_i=\sum_{\alpha'}p^{k-1}(\alpha') \rho_{\alpha'}\otimes\rho
\end{align}
but we know that for $n=k-1$ the statement is true, so
\begin{equation}
\mathbb{E}[\rho_{\alpha}]=\rho^{\otimes k-1}\otimes \rho =\rho^{\otimes k}
\end{equation}
So the statement is true for all $n \in \mathbb{N}$.
\end{proof}

If we have a projection, like $\Pi_{\alpha}$ then we can define a subspace which this projection projects to $\pi_{\alpha}=Im(\Pi_{\alpha})$. And vice versa, if we define a subspace $\pi_{\alpha}$, then we can define an orthogonal projection which project to this subspace $\Pi_{\alpha}$. This will be done throughout the paper by denoting with the same letter, indices the lowercase denotes the subspace the uppercase denotes the projections.

Consider the lattice of the projections. For 2 projection $P_1$ and $P_2$ denote $P_1 \vee P_2$ the projection which is the result of the $\vee$ operation in the net of projections (this means that $P_1\vee P_2$ is the projection which project to the subspace spanned by the range of $P_1$ and $P_2$). Similarly meaning has $P_1 \wedge P_2$ ($P_1\wedge P_2$ s the projection which projects to the subspace which is the section of the range of $P_1$ and $P_2$).

\begin{Lem}\label{tip-proj}
For every density matrices $\rho,\sigma$ and for every $\varepsilon>0$ 
there exist a projection $\Pi$, with properties
\begin{align}
	\Tr(\Pi \rho^{\otimes n} \Pi )\geq& 1-\varepsilon \label{t/1}\\
	\|\Pi \rho^{\otimes n} \Pi \| \leq& 2^{-n(\S(\rho)-\varepsilon)} \label{hiv1}\\
	\Tr(\Pi \sigma^{\otimes n} \Pi) \leq& 2^{-n(D(\rho|\sigma)-\varepsilon)} \label{t/3}
\end{align}
if $n>N(\rho,\sigma,\varepsilon)$, where $D(\rho|\sigma)$ denotes the quantum relative entropy of $\rho,\sigma$.
\end{Lem}
for proof see the appendix.
\begin{remark}
Note that, $\Pi$ does not depend on $\alpha_i$ the randomly chosen sequence, but $\tilde\Pi_{i-}$ does (across the article if an amount depend on the randomly generated sequence, then it will denoted by an $\alpha_i$ in the index or in the argument). This means that $\mathbb{E}(\Pi \rho_{\alpha_i}\Pi)=\Pi \rho^{\otimes n} \Pi$ while $\mathbb{E}(\rho_{\alpha_i})=\rho^{\otimes n}$, but $\mathbb{E}(\tilde\Pi_{i-} \rho_{\alpha_i}\tilde\Pi_{i-})\neq \tilde\Pi_{i-} \rho^{\otimes n} \tilde\Pi_{i-}$.
\end{remark}

\section{Reaching the Holevo bound with von Neumann measurement}\label{kodolas}

\begin{subsection}{Coding/Decoding}
First we generate $M = 2^{nR}$ random codewords with distribution $P$. 
These codewords are denoted by $\alpha_1, \alpha_2, \dots, \alpha_M$ and 
both the sender, and receiver are familiar with them. From these we generate 
quantum codewords.

A quantum codeword is a tensor product density defined in the following way:
If $\alpha_i(j) = k$, this is the $j$th symbol of the $i$th codeword, 
then the $j$th density of the $i$th tensor product is $\rrho_k$. The 
coding is as usual, we choose uniformly from the message set $\iM$ 
whose size is $M$ -- suppose this chosen message is $i$ --, and we send 
(or generate) the above defined quantum codeword for this message 
$\rrho_{\alpha_i(1)},\rrho_{\alpha_i(2)},\dots,\rrho_{\alpha_i(n)}$.
The quantum sequences go through the channel, the receiver gets 
the quantum sequence $\rho_{\alpha_i(1)},\rho_{\alpha_i(2)},\dots , 
\rho_{\alpha_i(n)}$ where $\rho_{\alpha_i(1)}=\iE(\rrho_{\alpha_i(1)})$.

Now we define a decoding algorithm which is nothing else than a POVM 
(von Neumann measurement). If the typical subspaces of $\rho_{\alpha_i}$ 
were orthogonal to each other, then - there would be no problem - we could 
make our POVM from $\Pi_{\alpha_i}$ (where $\Pi_{\alpha_i}$ defined in (\ref{hiv2})). However, this is not the case in general, 
and we have to orthogonalize them. We do this by a method, very similar to the 
Gram-Schmidt orthogonalisation method. 
In the first typical subspace let 
\begin{equation}
 \tilde\pi_{1}=span\{|\tilde s_{\alpha_1,k}\rangle: |\tilde s_{\alpha_1,k}\rangle=\Pi |s_{\alpha_1,k}\rangle ,\, |s_{\alpha_1,k}\rangle \in \pi_{\alpha_1}\}\label{hiv3}
\end{equation}
where $span\{\}$ means the subspace spanned by the the vectors in the curly bracket.

And for the $i$-th typical projections let
\begin{equation}
 \tilde\pi_{i}=span\left\{|\tilde s_{\alpha_1,k}\rangle: |\tilde s_{\alpha_1,k}\rangle=\Pi |s_{\alpha_i,k}\rangle-\sum_{j=1}^{i-1}\tilde\Pi_{j}\Pi |s_{\alpha_i,k}\rangle ,\, |s_{\alpha_i,k}\rangle \in \pi_{\alpha_i}\right\}\label{hiv4}
\end{equation}
(for defining subspaces to projections, and vice versa see Section \ref{Jel})

At the end of the procedure for all $s, t$ $\tilde\Pi_s$ is orthogonal to $\tilde\Pi_t$ $s \neq t$. We prove this, by total induction on $s,t$. Suppose that $s<t$. We see that for $s=1, t=2$ this is true because 
\begin{equation}
\tilde\Pi_1|\tilde s_{\alpha_2}\rangle=\tilde\Pi_1\Pi|s_{\alpha_2}\rangle - \tilde\Pi_1\tilde\Pi_1\Pi|s_{\alpha_2}\rangle=\tilde\Pi_1\Pi|s_{\alpha_2}\rangle - \tilde\Pi_1\Pi|s_{\alpha_2}\rangle=0 
\end{equation}
so $\tilde\pi_2\in Ker(\tilde\Pi_1)$.
 Suppose that, for all pairs $s',t'$, where $s'<s, t'<t$ the statement is true. then 
\begin{equation}
 \tilde\Pi_s|\tilde s_{\alpha_t}\rangle=\tilde\Pi_s\Pi|s_{\alpha_t}\rangle - \tilde\Pi_s\sum_{t'=1}^{t-1}\tilde\Pi_{t'}\Pi|s_{\alpha_2}\rangle=\tilde\Pi_s\Pi|s_{\alpha_t}\rangle - \tilde\Pi_s\Pi|s_{\alpha_t}\rangle=\underline 0 
\end{equation}
 because $\tilde\Pi_s\tilde\Pi_{t'}=0$ for $t'\neq s$ by the induction assumption.
Moreover the $\tilde\Pi_i$ and $\Pi$ are commutable operators (because every vector which spans the subspace of $\tilde\Pi_i$ projects onto is a member of $\pi$).

Our POVM (or our Von-Neumann measure) states from $\tilde\Pi_i$, plus we can make it complete adding an element of the POVM (labelled "error") on the remaining orthogonal subspace, if necessary, call these projection as $\tilde\Pi_{M+1}$. For index $i$, we can define the typical subspace of lesser indices as 
\[\tilde\Pi_{i-}=\sum_{j=1}^{i-1} \tilde\Pi_j\]

\end{subsection}

\begin{subsection}{The error probability of Decoding}
Now we show that the error probability is going to 0 if $n$ the block length goes to infinity. 
\begin{Tet}\label{basethm}
With these scheme, if $R < \chi(\rho)$ then for any $\gamma$ we can give a number $N$ such that if the length of the quantum codeword $n$ is longer than this number $n>N$ then average error probability is smaller than $\gamma$, provided that the blocklength $n$ is greater than $n_0(R,\gamma)$
\end{Tet}
To simplify the proof we use the following lemma.
\begin{Lem}\label{proj-lem}
The length of the projection of $|s_{\alpha_i,j}\rangle$ to $\tilde\Pi_i$ (element of our measuring POVM) can be underestimated - by two component where one is the projection to $\Pi$, and the other is orthogonal to the typical subspace of the lesser indices ($\tilde\Pi_{i-}$) - as follows
\begin{equation}\Tr(\tilde\Pi_i \rho_{\alpha_i} \tilde\Pi_i)\geq \left( \Tr (\Pi \tilde\rho_{\alpha_i} \Pi) -  \Tr ( \tilde\Pi_{i-} \Pi \tilde\rho_{\alpha_i}\Pi\, \tilde\Pi_{i-1}) \right)^2\end{equation}
(For definition of $\tilde\rho_{\alpha_i}$ see (\ref{hiv0}))
\end{Lem}
\begin{proof}
\begin{equation}\Tr(\tilde\Pi_i \rho_{\alpha_i} \tilde\Pi_i)=\Tr(\tilde\Pi_i\rho_{\alpha_i} )=\Tr(\tilde\Pi_i \sum_{j=1}^{d}\lambda_{\alpha_i,j} |s_{\alpha_i,j}\rangle\langle s_{\alpha_i,j}|)=\end{equation}
Now we decompose the projection $\tilde\Pi_i$ depend on $j$ as follows
\begin{equation}\tilde\Pi_i=a^2_j|\tilde s_{\alpha_i,j}\rangle \langle \tilde s_{\alpha_i,j}| + \hat{\Pi_i}\end{equation}
where $a_j>1$ is the reciprocal of the length of $|\tilde s_{\alpha_i,j}\rangle =\Pi |s_{\alpha_i,k}\rangle-\sum_{j=1}^{i-1}\tilde\Pi_{j}\Pi |s_{\alpha_i,k}\rangle$ (length of $|\tilde s_{\alpha_i,j}\rangle$ is smaller than 1 because it is an orthogonal component of a unit length vector $|s_{\alpha_i,j}\rangle$). So we continue the first row by
\begin{align}
&=\sum_{j=1}^{d}\lambda_{j} \Tr( a^2_j|\tilde s_{\alpha_i,j}\rangle \langle \tilde s_{\alpha_i,j}||s_{\alpha_i,j}\rangle\langle s_{\alpha_i,j}| + \hat\Pi_i |s_{\alpha_i,j}\rangle\langle s_{\alpha_i,j}|)=\\
&\sum_{j=1}^{d}\lambda_{j} \Tr( a^2_j|\tilde s_{\alpha_i,j}\rangle \langle \tilde s_{\alpha_i,j}||s_{\alpha_i,j}\rangle\langle s_{\alpha_i,j}|)+
\sum_{j=1}^{d}\lambda_{j} \Tr(\langle s_{\alpha_i,j}| \hat\Pi_i |s_{\alpha_i,j}\rangle)\geq  \\
&\geq \sum_{j=1}^{d(\alpha_i)}\lambda_{j} \Tr(|\tilde s_{\alpha_i,j}\rangle \langle \tilde s_{\alpha_i,j}||s_{\alpha_i,j}\rangle\langle s_{\alpha_i,j}|)= \label{kifgasolt2}\\
&=\sum_{j=1}^{d(\alpha_i)}\lambda_{j} \Tr((\Pi |s_{\alpha_i,j}\rangle-\tilde\Pi_{i-}\Pi| s_{\alpha,j}\rangle) (\langle s_{\alpha_i,j}|\Pi -\langle s_{\alpha_i,j}|\Pi\, \tilde\Pi_{i-} )|s_{\alpha_i,j}\rangle\langle s_{\alpha_i,j}|)=\\
&=\sum_{j=1}^{d(\alpha_i)}\lambda_{j} \left(\langle s _{\alpha_i,j}|(\Pi |s_{\alpha_i,j}\rangle-\Pi_{i-}\Pi| s_{\alpha,j}\rangle) \right)^2
\end{align}
Where (\ref{kifgasolt2}) comes from the fact, that $a_j>1$ so $a_j^2>1$. From the Jensen's inequality
\begin{align}
&\geq \left(\sum_{j=1}^{d(\alpha_i)}\lambda_{j} \langle s _{\alpha_i,j}|(\Pi |s_{\alpha_i,j}\rangle-\tilde\Pi_{i-}\Pi| s_{\alpha,j}\rangle) \right)^2=\\
&=\left(\sum_{j=1}^{d(\alpha_i)}\lambda_{j} \langle s_{\alpha_i,j}|\Pi |s_{\alpha_i,j}\rangle - \sum_{j=1}^{d(\alpha_i)}\lambda_{j} \langle s_{\alpha_i,j}| \tilde\Pi_{i-} \Pi |s_{\alpha_i,j}\rangle \right)^2=\\
&=\left(\sum_{j=1}^{d(\alpha_i)}\lambda_{j} \Tr (\Pi |s_{\alpha_i,j}\rangle \langle s_{\alpha_i,j}| - \sum_{j=1}^{d(\alpha_i)}\lambda_{j} \Tr (\tilde\Pi_{i-} \Pi |s_{\alpha_i,j}\rangle \langle s_{\alpha_i,j}|) \right)^2=\\
&=\left( \Tr (\Pi \tilde\rho_{\alpha_i}) -  \Tr (\tilde\Pi_{i-} \Pi  \tilde\rho_{\alpha_i}) \right)^2
\geq \left( \Tr (\Pi \tilde\rho_{\alpha_i} \Pi) -  \Tr ( \tilde\Pi_{i-} \Pi \tilde\rho_{\alpha_i} \Pi\,\tilde\Pi_{i-}) \right)^2
\end{align}
\end{proof}
Now see the proof: \\
\begin{proof}[Proof of Theorem \ref{basethm}]
Let $\varepsilon$ is such small that $R+\varepsilon < \chi(\rho) $, and $8\varepsilon\leq \gamma$ and let $n$ such large that all the Lemmas and definitions from Section \ref{Jel} with $\varepsilon$ are true.
\begin{align}
P_e&=\mathbb{E}[1-\frac{1}{M}\sum_{i=1}^M \Tr(\tilde\Pi_i\rho_{\alpha_i}\tilde\Pi_i)]\leq\\
&\leq \frac{1}{M}\sum_{i=1}^M \mathbb{E}[1-\Tr(\tilde\Pi_i\rho_{\alpha_i} \tilde\Pi_i)]\leq \\
&\leq \frac{1}{M}\sum_{i=1}^M \mathbb{E}[1-\Tr(\tilde\Pi_i\rho_{\alpha_1} \tilde\Pi_i)]
\end{align}
We use the Lemma
\begin{align}
P_e\leq \frac{1}{M}\sum_{i=1}^M 1-\mathbb{E}[\Tr(\Pi \tilde\rho_{\alpha_i} \Pi - \tilde\Pi_{i-} \Pi \tilde\rho_{\alpha_i} \Pi \tilde\Pi_{i-} )]^2
\end{align}
We use the Jensen's inequality
\begin{align}
P_e\leq \frac{1}{M}\sum_{i=1}^M 1-[\mathbb{E}\Tr(\Pi \tilde\rho_{\alpha_i} \Pi - \tilde\Pi_{i-} \Pi \tilde\rho_{\alpha_i} \Pi \tilde\Pi_{i-} )]^2=\\
=\frac{1}{M}\sum_{i=1}^M 1-[\Tr(\Pi \mathbb{E}(\tilde\rho_{\alpha_i}) \Pi -  \tilde\Pi_{i-} \Pi \mathbb{E}(\tilde\rho_{\alpha_i}) Pi \tilde\Pi_{i-} )]^2=\\
\frac{1}{M}\sum_{i=1}^M 1-[\Tr(\Pi \tilde\rho \Pi) - \mathbb{E} (\Tr( \tilde\Pi_{i-} \Pi \tilde\rho_{\alpha_i} \Pi \tilde\Pi_{i-} ))]^2
\end{align}
From (\ref{kifgasolt3}), and remark \ref{kifgasolt4}, we know that $\Tr(\tilde\rho)\geq 1-2\varepsilon$ and $\rho\geq \tilde\rho$ so, for $\Delta=\rho-\tilde\rho$, $\Tr(\Delta)\leq 2\varepsilon$. With this
\begin{equation}\Tr(\Pi \tilde\rho \Pi)=\Tr(\Pi \rho \Pi)- \Tr (\Pi (\rho-\tilde\rho) \Pi)\geq 1-\varepsilon- \Tr(\Delta) \label{hullamisnagy} \end{equation}
So the error probability
\begin{align}
P_e\leq \frac{1}{M}\sum_{i=1}^M 1-[1-3\varepsilon - \mathbb{E}(\Tr( \tilde\Pi_{i-} \Pi \tilde\rho_{\alpha_i} \Pi \tilde\Pi_{i-} ))]^2\leq\\
\leq \frac{1}{M}\sum_{i=1}^M 1-[1-3\varepsilon -  \mathbb{E}(\Tr(\tilde\Pi_{i-} \Pi \tilde\rho_{\alpha_i} \Pi \tilde\Pi_{i-} ))]^2
\end{align}
Now we analyse the last term
\begin{align}
\mathbb{E}(\Tr( \tilde\Pi_{i-} \Pi \tilde\rho_{\alpha_i} \Pi \tilde\Pi_{i-} ))=\mathbb{E}(\Tr( \tilde\Pi_{i-} \Pi \tilde\rho_{\alpha_i} \Pi ))=\mathbb{E}( \Tr(\sum_{j=1}^{dim(\tilde\Pi_{i-})}|b_
j\rangle \langle b_j| \Pi \tilde\rho_{\alpha_i} \Pi ))=\\
\mathbb{E}(\sum_{i=1}^{dim(\tilde\Pi_{i-})} \langle b_j|\Pi \tilde\rho_{\alpha_i} \Pi |b_j\rangle)
\end{align}
Now we have to evaluate $dim(\tilde\Pi_{i-})$. Because of (\ref{kifgasolt1})
\begin{equation}dim(\tilde\Pi_{i-})\leq \sum_{j=2}^M dim(\Pi_{\alpha_j})\leq M2^{nS(\rho|\alpha)}=2^{n (R+S(\rho|\alpha))}\end{equation}
\begin{align}
\mathbb{E}(\Tr( \tilde\Pi_{i-} \Pi \tilde\rho_{\alpha_i} \Pi \tilde\Pi_{i-} ))\leq 2^{n (R+\S(\rho|\alpha))}\max_j \mathbb{E}(\langle b_j|\Pi \tilde\rho_{\alpha_i} \Pi |b_j\rangle )=\\
2^{n (R+\S(\rho|\alpha))}\max_j \langle b_j|\Pi \tilde\rho \Pi |b_j\rangle \leq 2^{n (R+\S(\rho|\alpha))}\max_j\langle  b_j|\Pi \rho \Pi |b_j\rangle
\end{align}
Because $b_j$'s are unit vectors, from (\ref{hiv1}) we know that $\max_j\langle  b_j|\Pi \rho \Pi |b_j\rangle\leq 2^{-n(\S(\rho)-\varepsilon)}$, so
\begin{equation}\mathbb{E}(\Tr( \tilde\Pi_{i-} \Pi \tilde\rho_{\alpha_i} \Pi \tilde\Pi_{i-} ))\leq 2^{n (R+\S(\rho|\alpha))}2^{-n(\S(\rho)-\varepsilon)}\leq 2^{-n(\S(\rho)-\S(\rho|\alpha)-R-\varepsilon)}\end{equation}
Now we can see that $\S(\rho)-\S(\rho|\alpha)=\chi(\rho)$, and we assumed that $R+\varepsilon <\chi(\rho)$, so the exponent is negative. If $n$ is large enough then the whole expression is less than $\varepsilon$.
\begin{equation}\mathbb{E}(\Tr( \tilde\Pi_{i-} \Pi \tilde\rho_{\alpha_i} \Pi \tilde\Pi_{i-} ))\leq \varepsilon\end{equation}
So the error probability is smaller than
\begin{equation}P_e\leq 1-[1-4\varepsilon]^2=8\varepsilon-16\varepsilon^2\leq 8\varepsilon\end{equation}
\end{proof}
\end{subsection}

\begin{section}{I. Use:\\ Coding for finite compound channel}\label{Fehaszn1}
First we give a definition of the compound channel. Suppose there is a given set of channels $\irott{S}$. We want a predefined coding scheme to code our message with the following disturbing effect: suppose there is an enemy who chose one channel from the set after we generated our quantum codeword. Now our quantum codeword are send through the chosen channel (so all the quantum has the same effect but we cannot say which).

This is a more realistic model than a simple quantum channel, we know what effects can destroy our quantum codewords,but we cannot know at the present moment, which effect is active.
This compound channel is a finite compound channel, if the set $\irott{S}$ is finite.
Define:
\begin{align}
\chi(\irott{S},P,\rrho_1^l)\circeq&\min_{\iE\in \irott{S}}\chi(\iE,P,\rrho_1^l)\\
\chi(\irott{S})\circeq&\max_{P,\rrho_1^l}\chi(\irott{S},P,\rrho_1^l)\label{compcap}
\end{align}
\begin{Tet}
The classical capacity $C$ of the finite compound channel is $\chi(\irott{S})$, This means, if $R<\chi(\irott{S})$ then for any $\gamma$ there exist a number $N(\gamma,R,|\irott{S}|)$ that if the length of the quantum codeword $n$ is larger than that number $n>N(\gamma)$ then the error probability is smaller than $\gamma$.
\end{Tet}
\begin{proof}
First we can assume that $|\irott S|>2$, because for $|\irott S|=1$, the proof is the same as theorem \ref{basethm}. The inequality $C \leq \chi(\irott{S})$ is simple. Because in \cite{Holevo} was shown, that a quantum system can carry $\chi(\iE,P,\rrho_1^l)$ bit information if the sender code his message to quantum states $\iE(\rrho_i)$ with a priori probabilities $P$. If the sender has no knowledge, which channel is being used, then the sender can only codes his message to states $\rrho_i$ with a priori distribution $P$. Then if his enemy chose the worst channel for these schema, the carried information cannot be larger than \[\max_{P,\rrho_1^l} \min_{\iE\in \irott{S}} \chi(\iE,P, \rrho_1^l) = \chi(\irott{S}).\] To prove that $C \geq \chi(\irott{S})$ we show a coding scheme which rate can reach the capacity.

\begin{subsection}{Coding for finite compound channel}
Let $|\irott{S}| = a$, $\irott{S}=\iE_1,\iE_2,\dots,\iE_a$ fixed. Let $P$, and $\rrho_1, \rrho_2, \dots, \rrho_l$ be the probability distribution and quantum's that reach the maximum in (\ref{compcap}). The sender codes his message to randomly generated codeword as in section \ref{kodolas}. The codeword goes into the quantum channel $\iE_o$ $o\in \{1,2,\dots,a\}$ which was chosen by our enemy.
Denote $\rrho=\sum_{j=1}^l p_l \rrho_l$ the input mixed state, and denote $\rho^k=\iE_k^{\otimes n}(\rrho^{\otimes n}),\, 1\leq k \leq a$ the possible mixed output of the channel. Similarly  denote $\rho^k_{\alpha_i}=\iE_k(\rho_{\alpha_i})$ the possible output of the $i$-th quantum codeword. To simplify our proof we can assume that, the order of the set of quantum channels is such that, the first $\bar{a}$ $\rho^k, k \in \{1,2,\dots,\bar{a}\}$ is different.
\end{subsection}
\begin{subsection}{Decoding for finite compound channel}
Decoding is done by two steps:
In first step we can detect which mixed state we have. In the second step we detect the message.

See the first step. Let denote $\bar{a}$ the number of the different output mixed states.
To distinguish the output mixed states we will use our Lemma  \ref{tip-proj}. Let $\varepsilon$ be such small that $\varepsilon< \min_{i\neq j 1\leq i,j\leq \bar{a} } \frac{D(\rho^i\|\rho^j)}{2}$, $\varepsilon\leq \frac{\gamma}{8a}$ and $R+\varepsilon<\chi(\irott{S})$. Then for an $n$ large enough for every $i,j: i\neq j 1\leq i,j \leq \bar{a}$ Lemma \ref{tip-proj} is true, with $\rho=\rho^l,\sigma=\rho^k$ and get $\Pi^{i,j}$ (if $n$ is greater than $\max(N(\rho^l,\rho^k,\varepsilon))$). Now we make a typical projection for all $i$ as follows:
\begin{equation}\Pi^i=\bigwedge_{1\leq j\leq \bar{a}} \Pi^{i,j}\end{equation}
See, that
\begin{align}
 \Tr(\Pi^i(\rho^i)^{\otimes n}\Pi^i)=&1-\Tr((\Pi^i)^c(\rho_i)^{\otimes n}(\Pi^i)^c)\geq\\
\geq& 1-\sum_{j\neq i,1\leq j \leq \bar{a}} \Tr((\Pi^{i,j})^c(\rho_i)^{\otimes n}(\Pi^{i,j})^c)\geq 1-a\varepsilon\\
 \|\Pi^i (\rho^i)^{\otimes n} \Pi^i\|\leq&2^{-n(S(\rho_i)-\varepsilon)}\\
 \Tr(\Pi^i (\rho^j)^{\otimes n} \Pi^i) \leq& 2^{-n(D(\rho^i\|\rho^j)-\varepsilon)}
\end{align}
for all $1 \leq i,j \leq a $.

We detect the mixed state as follows:
First we give a sequence of measures. Our first POVM measure states from $\Pi^1,(\Pi^1)^c$ if we measure $\Pi^1$ we know, that our mixed state was $\rho^1$ so we stop, if we measured $(\Pi^1)^c$ then we measure again. Second POVM measure states from $\Pi^2,(\Pi^2)^c$, etc. With this we can differentiate the possible $\bar{a}$ type of our mixed state.

We suppose that our enemy chosen the channel $\iE_o$ which generate the $k$-th mixed state (this means $\iE_l(
\rrho)=\iE_o(\rrho)$), if $k \neq 1$ then the first measure is good, if we measure the second projection. This probability is
\begin{equation}\Tr((\Pi^1)^c \rho^k_{\alpha_i} (\Pi^1)^c)\end{equation}
and our quantum states will be
\begin{equation}\frac{(\Pi^1)^c \rho^k_{\alpha_i} (\Pi^1)^c}{\Tr((\Pi^1)^c \rho^k_{\alpha_i} (\Pi^1)^c)}\end{equation}
Next if $k\neq 2$ then the next measure is good if we measure the second projection this probability is
\begin{equation}\Tr((\Pi^2)^c\frac{(\Pi^1)^c \rho^k_{\alpha_i} (\Pi^1)^c}{\Tr((\Pi^1)^c \rho^k_{\alpha_i} (\Pi^1)^c)}(\Pi^2)^c)\end{equation}
and our state becomes
\begin{equation}\frac{(\Pi^2)^c\frac{(\Pi^1)^c \rho^k_{\alpha_i} (\Pi^1)^c}{\Tr((\Pi^1)^c \rho^k_{\alpha_i} (\Pi^1)^c)}(\Pi^2)^c}{\Tr((\Pi^2)^c\frac{(\Pi^1)^c \rho^k_{\alpha_i} (\Pi^1)^c}{\Tr((\Pi^1)^c \rho^k_{\alpha_i} (\Pi^1)^c)}(\Pi^2)^c)}=\frac{(\Pi^2)^c(\Pi^1)^c \rho^k_{\alpha_i} (\Pi^1)^c(\Pi^2)^c}{\Tr(\Pi^2)^c(\Pi^1)^c \rho^k_{\alpha_i} (\Pi^1)^c(\Pi^2)^c)}\end{equation}
And the probability, that we don't made error through the first, and the second step is
\begin{align}\Tr\left(\frac{(\Pi^2)^c(\Pi^1)^c \rho^k_{\alpha_i} (\Pi^1)^c(\Pi^2)^c}{\Tr((\Pi^1)^c \rho^k_{\alpha_i} (\Pi^1)^c)}\right)&\Tr((\Pi^1)^c \rho^k_{\alpha_i} (\Pi^1)^c)=\\
 &\Tr((\Pi^2)^c(\Pi^1)^c \rho^k_{\alpha_i} (\Pi^1)^c(\Pi^2)^c)
\end{align}
From this we can see that, if $\iE_o(\rrho)=\iE_k(\rrho)$ then the probability that we detect our mixed state correctly is:
\begin{align}
 \mathbb{E}[\Tr(\Pi^k (\Pi^{k-1})^c\dots (\Pi^{2})^c (\Pi^{1})^c  \rho^k_{\alpha_i} (\Pi^{1})^c (\Pi^{2})^c \dots (\Pi^{k-1})^c \Pi^k)]
 \end{align}
Let $\delta=a\varepsilon$. From Lemma \ref{eri=r}, and from the definition of $\Pi^k$ we know
\begin{align}
 \mathbb{E}[\Tr(\Pi^k  \rho^k_{\alpha_i}  \Pi^k)] =\Tr(\Pi^k  \rho^k  \Pi^k)\geq 1-\delta
\end{align}
Moreover
\begin{align}
 1-\delta &\leq \Tr(\Pi^k  \rho^k  \Pi^k)\leq \Tr(\Pi^k (\Pi^{1})^c \rho^k (\Pi^{1})^c \Pi^k)+ \Tr(\Pi^k \Pi^{1} \rho^k \Pi^{1} \Pi^k)\leq \\
&\leq \Tr(\Pi^k (\Pi^{1})^c \rho^k (\Pi^{1})^c \Pi^k)+ \Tr(\Pi^{1} \rho^k \Pi^{1} )\\
&\leq \Tr(\Pi^k (\Pi^{1})^c \rho^k (\Pi^{1})^c \Pi^k)+ 2^{-nD(\rho^{1}\|\rho^k)-\varepsilon}=\\
&=\Tr(\Pi^k (\Pi^{2})^c (\Pi^{1})^c  \rho^k (\Pi^{1})^c (\Pi^{2})^c \Pi^k)+ \notag\\
&+\Tr(\Pi^k \Pi^{2} (\Pi^{1})^c \rho^k (\Pi^{1})^c \Pi^{2} \Pi^k)+2^{-nS(\rho^{1}|\rho^k)-\varepsilon}\leq\\
&\leq \Tr(\Pi^k (\Pi^{2})^c (\Pi^{1})^c  \rho^k (\Pi^{1})^c (\Pi^{2})^c \Pi^k)+\Tr( \Pi^{2} \rho^k \Pi^{2} )+2^{-nD(\rho^{1}\|\rho^k)-\varepsilon}\leq \\
&\leq \Tr(\Pi^k (\Pi^{2})^c (\Pi^{1})^c  \rho^k (\Pi^{1})^c (\Pi^{2})^c \Pi^k)+2^{-nD(\rho^{2}\|\rho^k)-\varepsilon}+2^{-nD(\rho^{1}\|\rho^k)-\varepsilon}\leq\\
& \leq  \dots \leq \Tr(\Pi^k (\Pi^{k-1})^c\dots (\Pi^{2})^c (\Pi^{1})^c  \rho^k_{\alpha_i} (\Pi^{1})^c (\Pi^{2})^c \dots (\Pi^{k-1})^c \Pi^k)+ \notag\\
&+\sum_{l=1}^{k-1} 2^{-nD(\rho^{l}\|\rho^k)-\varepsilon}
\end{align}
This means that
\begin{align}
& 1-\delta-\sum_{l=1}^{k-1} 2^{-nD(\rho^{l}\|\rho^k)-\varepsilon}\\
& \leq \mathbb{E}[\Tr(\Pi^k (\Pi^{k-1})^c\dots (\Pi^{2})^c (\Pi^{1})^c  \rho^k_{\alpha_i} (\Pi^{1})^c (\Pi^{2})^c \dots (\Pi^{k-1})^c \Pi^k)]
\end{align}
See, that
\begin{align}
1-\delta-\bar{a}\min_{1\leq l<k} 2^{-nD(\rho^{l}\|\rho^k)-\varepsilon} \leq 1-\delta-\sum_{l=1}^{k-1} 2^{-nD(\rho^{l}\|\rho^k)-\varepsilon}
\end{align}
Because $\varepsilon< \frac{S(\rho_i|\rho_j)}{2} $ this means that if $n$ is enough large then
\begin{equation}\bar{a}\min_{l\in \{1,2,\dots,k\}} 2^{-nS(\rho^{l}|\rho^k)-\varepsilon} \leq \delta\end{equation}
which means that the expectation value of the probability of the good detection is greater than $1-2\delta$. This mean the the error (that we detect a wrong mixed state, or all the measure never decide for the first projection) is smaller than $2\delta$. See that this bound is valid  for all possible $1\leq k \leq \bar{a}$. Denote by $P^k$ the following operator $\Pi^k (\Pi^{k-1})^c\dots (\Pi^{2})^c (\Pi^{1})^c$, with this notation at the end of the procedure, our quantum codeword $\rho_{\alpha_i}$ will be in the form $\frac{P^k \rho^o_{\alpha_i} {P^k}^*}{\Tr(P^k \rho^o_{\alpha_i} {P^k}^*)}$

Let see the second step, now we detect the message. Suppose that we detected that our mixed state is $\rho^k$ which mixed state can be generated by $\iE^{k_1},\iE^{k_2},\dots,\iE^{k_l}$ $1 \leq k_j \leq a, l\leq a$ and we know that our enemy chosen $\iE^o$ so some $k_j=o$. Prepare all $\Pi^{k_j}_{\alpha_i}$ as in Section \ref{kodolas}. Now define
\begin{equation}\check\pi_{i}=span\{\bigcup_{l:\iE^l(\rrho)=\iE^o(\rrho)}\pi^l_{i}\}\end{equation}
where $\pi^l_{i}$ is the typical subspace (\ref{hiv2}) of $\rho^l_{\alpha_i}=\iE_l(\rrho_{\alpha_i})$ And define the typical projections for the message as in (\ref{hiv3}), (\ref{hiv4})
\begin{equation}\hat\pi_{i}=span\left\{P^{k} |s\rangle - \sum_{j=1}^{i-1}\hat\Pi_{j} P^{k}|s\rangle ,|s\rangle \in \check\pi_{i} \right\}
\end{equation}
We made our POVM as in Section \ref{kodolas}, from these orthogonal projection, with a possible complement with an error labelled subspace. Similarly to Section \ref{kodolas}, we define $\hat\Pi_{i-}=\sum_{j=1}^{i-1}\hat\Pi_j$ 

For these measurement a similar statement is true as in Lemma \ref{proj-lem}
\begin{equation}\Tr(\hat{\Pi}_{i} P^k \rho^o_{\alpha_i} {P^k}^* \hat{\Pi}_{i}) \geq \left( \Tr (P^k \tilde\rho^{o}_{\alpha_i} {P^k}^*) -  \Tr (\hat\Pi_{i-} P^k\, \tilde\rho^o_{\alpha_i}{P^k}^*\hat\Pi_{i-}) \right)^2 \end{equation}
The proof is exactly the same as Lemma \ref{proj-lem}.
So we can calculate the error probability of the message detection (with the good ''mixed`` state detection):
\begin{align}
P_e&=\mathbb{E}[\Tr(P^k \rho^o_{\alpha_i} {P^k}^*) \frac{1}{M}\sum_{i=1}^M 1-\Tr(\hat{\Pi}_i \frac{P^k \rho^o_{\alpha_i} {P^k}^*}{\Tr(P^k \rho^o_{\alpha_i} {P^k}^*)} \hat{\Pi}_i)]\leq \\ &\leq\sum_{i=1}^M\frac{ \mathbb{E}[1-\Tr(\hat{\Pi}_i P^k \rho^o_{\alpha_i} {P^k}^* \hat{\Pi}_i)]}{M}\notag
\end{align}
with the previous statement, and with the Jensen's inequality
\begin{align}
P_e&\leq \frac{1}{M}\sum_{i=1}^M\mathbb{E}[1-\left( \Tr (P^k \tilde\rho^o_{\alpha_i} {P^k}^*) -  \Tr (\hat\Pi_{i-} P^k \tilde\rho^o_{\alpha_i}{P^k}^* \hat\Pi_{i-}) \right)^2]\leq\\
&\leq \frac{1}{M}\sum_{i=1}^M\left[1-\left( \Tr (P^k \tilde\rho^o {P^k}^*) -  \mathbb{E}[\Tr (\hat\Pi_{i-} P^k  \tilde\rho^{o}_{\alpha_i}{P^k}^* \hat\Pi_{i-})] \right)^2\right]
\end{align}
Because $P^k \rho^o {P^k}^*\geq 1-2\delta$ and $\Tr(\rho^o-\tilde\rho^o)\geq 1-\delta$ the first term is greater than $1-3\delta$ as in proof of Theorem \ref{basethm}. The second term is
\begin{align}
\mathbb{E}(\Tr( \hat\Pi_{i-} P^k \tilde\rho^o_{\alpha_i} {P^k}^* \hat\Pi_{i-} ))=\mathbb{E}(\Tr(\hat\Pi_{i-} P^k \tilde\rho^o_{\alpha_i} {P^k}^* ))=\\
\mathbb{E}( \Tr(\hspace{-10pt} \sum_{j=1}^{dim(\hat\Pi_{i-})}\hspace{-10pt}|b_
j\rangle \langle b_j| P^k \tilde\rho^o_{\alpha_i} {P^k}^* ))\leq
 \mathbb{E}(\hspace{-10pt}\sum_{j=1}^{dim(\hat\Pi_{i-})}\hspace{-10pt} \langle b_j|\Pi^k \tilde\rho^o_{\alpha_i} \Pi^k |b_j\rangle)\leq dim(\hat\Pi_{i-})\|\Pi^k \tilde\rho^k \Pi^k\|
\end{align}
because $P^k \tilde\rho^o_{\alpha_i} {P^k}^* \leq \Pi^k \tilde\rho^o_{\alpha_i} \Pi^k$, and we assumed that $\iE_o(\rrho)=\iE_k(\rrho)$. Now we have to evaluate $dim(\hat\Pi_{i-})$
\begin{equation}dim(\hat\Pi_{i-})\leq \sum_{j=1}^l\sum_{s=i+1}^M dim(\Pi^{k_j}_{\alpha_s})\leq aM2^{nS(\rho|\alpha)}=a2^{n (R+S(\rho|\alpha))}\end{equation}
We use (\ref{hiv1}) and get
\begin{align}
\mathbb{E}(\Tr( \hat\Pi_{i-} \Pi^k \tilde\rho^o_{\alpha_i} \Pi^k \hat\Pi_{i-} ))\leq
a 2^{-n[S(\rho^o)-S(\rho^o|x)-R-\varepsilon]}
\end{align}
which is smaller than $\delta$ if $n$ is large enough, so
\begin{equation}P_e\leq \frac{1}{M}\sum_{i=1}^M [1-(1-3\delta-\delta)^2<8\delta\leq 8a\varepsilon \end{equation}
Because $\varepsilon< \frac{\gamma}{8a}$ with this the theorem is proved.
\end{subsection}
\end{proof}
\end{section}

\begin{section}{II. Use: Practical considerations}
One could think that, after the articles of Schumacher or Holevo \cite{Schumacher}, \cite{Holevo} that we can communicate classical data through a quantum channel optimally. However this is true only in theory, because to measure a POVM with many output (the needed output of the POVM grows exponentially in $n$ the length of the codeword) is very difficult in practice. But as we will see, this is not the case in the von Neumann measurement, we will give a detection algorithm - a sequence of measure - where the number of outcomes of the measures are always 2.

We introduce the following notation for $1\leq i < j \leq M+1$
\begin{equation}D_{\{i,j\}}=\tilde\Pi_i + \tilde\Pi_{i+1} + \dots + \tilde\Pi_{j}\end{equation}
For simplicity, suppose that $M+1=2^k$. Then the detection algorithm can be the following:
First we measure a von Neumann measurement states from 
$D_{\{1,(M+1)/2\}},D_{\{(M+1)/2+1,M+1\}}$.
In every measurement, if the result is the first operator, we give 0, if the second we give 1.

Now we measure again. Of course on a quantum state that is modified by the previous measurement.
In each next step we half the interval of the previous measurement. If our measurement gave the $\{i,j\}$ our measurement will states from $D_{\{i,i+(j-1)/2\}}$,$D_{\{i+(j-1)/2+1,j\}}$.
For example the second step looks like follows: If we measured 0 then our measurement will states from $D_{\{0,\dots,(M+1)/4\}},D_{\{(M+1)/4+1,\dots,(M+1)/2\}}$, if the previous measurement had gave the result 1 then our measurement will states from :\\$D_{\{(M+1)/2+1,\dots,3(M+1)/4\}},D_{\{3(M+1)/4+1,\dots,(M+1)\}}$.
At the end the 0-s and 1-s give the number of the message in binary form. If we get only 1-s then we declare error.

See that the probability of the good detection not changes. Suppose that we send the first message, then the first measure will give the good result with probability
\begin{equation}\Tr(D_{\{\{1,(M+1)/2\}}\rho_{\alpha_1}D_{\{1,(M+1)/2\}})\end{equation}
and will the state will change to
\begin{equation}\frac{D_{\{1,(M+1)/2\}}\rho_{\alpha_1}D_{\{1,(M+1)/2\}}}{\Tr(D_{\{1,(M+1)/2\}}\rho_{\alpha_1}D_{\{1,(M+1)/2\}})}\end{equation}
The second measurement will be good with probability
\begin{equation}\Tr\left( D_{\{1,(M+1)/4\}} \frac{D_{\{1,(M+1)/2\}}\rho_{\alpha_1}D_{\{1,(M+1)/2\}}}{\Tr(D_{\{1,(M+1)/2\}}\rho_{\alpha_1}D_{\{1,(M+1)/2\}})} D_{\{1,(M+1)/4\}} \right) \end{equation}
But $D_{\{1,(M+1)/4\}} < D_{\{1,(M+1)/2\}}$ so this simplify to
\begin{equation}\Tr\left(\frac{D_{\{1,(M+1)/4\}}\rho_{\alpha_1}D_{\{1,(M+1)/4\}}}{\Tr(D_{\{1,(M+1)/2\}}\rho_{\alpha_1}D_{\{1,(M+1)/2\}})}\right)\end{equation}
And the state change to
\begin{equation}\frac{D_{\{1,(M+1)/4\}}\rho_{\alpha_1}D_{\{1,(M+1)/4\}}}{\Tr(D_{\{1,(M+1)/4\}}\rho_{\alpha_1}D_{\{1,(M+1)/4\}})}\end{equation}
So the probability that the first two measurement was true is
\begin{align}\Tr(D_{\{\{1,(M+1)/2\}}\rho_{\alpha_1}D_{\{1,(M+1)/2\}}) \Tr\left(\frac{D_{\{1,(M+1)/4\}}\rho_{\alpha_1}D_{\{1,(M+1)/4\}}}{\Tr(D_{\{1,(M+1)/2\}}\rho_{\alpha_1}D_{\{1,(M+1)/2\}})}\right) =\\
\Tr(D_{\{1,(M+1)/4\}}\rho_{\alpha_1}D_{\{1,(M+1)/4\}})
\end{align}
with keep going this train of thought we can see, that at the end that the probability that all of the measurement was good is not else, than
\begin{equation}\Tr(D_{i,i}\rho_{\alpha_i}D_{i,i})=\Tr(\tilde\Pi_i\rho_{\alpha_i}\tilde\Pi_i)\end{equation}
which the same as the error probability in Section \ref{kodolas}. Of course this procedure can be generalised to case when we have finite possible outcomes.

This means that it is possible to classically code/decode classical information with quantum apparatus in an optimal way in linear time. This is a quite surprising result, because in classical information theory to reach the Shannon limit in polynomial time is an unresolved problem (the best result needs $n^{\log(n)}$ time). Usually the classical information theory is considered as a part of quantum information theory, which would mean that optimal decoding of classical channel in linear time is possible. This means that if a quantum machine can perform arbitrarily von Neumann measurement with only two possible outcome, then this machine can solve non-polynomial classical problems in linear time.

Now we will show how can be a classical message through a quantum apparatus decoded. Suppose that there is a classical setup with a discrete memoryless channel. There is a given state transition matrix $W(y|x)$ (with input output alphabet $1,2,\dots,l$ $1,2,\dots,d$ ) and a given optimal input distribution $P$.
Now model the classical system with a quantum one. Let define for each $x$ $\rho_x=diag(W(\cdot|x))$ (where $diag(W(\cdot|x))$ denotes a diagonal matrix we get from the output distribution provided by $x$ in another form $\sum_{a \in \{1,\dots,d\}} W(a|x)E_{a,a}$ ). In these case all the classical and all the quantum information quantities are equivalent ($\chi(\iE)=C(W)$, $H(\cdot)=S(\cdot)$). Then we know from Section \ref{kodolas} that there exist $2^{Rn}$ piece of sequence that with $\rho_{\alpha_i}$ quantum codewords we can optimally communicate.  Compute the optimal von Neumann measurement as in Section \ref{kodolas}. Now we use the $\alpha_i$ sequence as an input codeword for our classical channel, and decode the classical channel as follows:
We get the classical signal, we coded into quantum sequences, we perform the measurement, after that we get the number of the message was sent, so we decoded the message (in linear time as in the beginning of these section).

Denote the output signal of $\alpha_j$ by $\Vy(j)$. Denote the $i$-th component of $\Vy(j)$ by $y_i(j)$
We get the signal $\Vy(j)$ and code every symbol of it, into a quantum in the following way:
\begin{equation}y_i(j) \rightarrow E_{y_i(j),y_i(j)}\end{equation}
which means if we get the first symbol of my output alphabet we code into a quantum represented by $E_{1,1}$, where $E_{i,j}$ denotes the matrix with 1 in the $i$-th row $j$-th column and 0 elsewhere.
Denote these quantum sequence by $\mu_{\Vy(j)}$
Now see that the error probability of the event that the $i$-th message was wrongly decoded:
\begin{equation}\mathbb{E}[1-\Tr(\mu_{\Vy(j)})]\end{equation}
We have to take the expectation value because the output sequence $\Vy(j)$ can varied. It can be easily proved that $\mathbb{E}[\mu_{\Vy(j)}]=\rho_{\alpha_j}$.
So the average error probability is same as in Section \ref{kodolas}. Which means that classical messages can be decoded by quantum apparatus in linear ($nR$) time.

\begin{proof}
 of $\mathbb{E}[\mu_{\Vy(j)}]=\rho_{\alpha_j}$
proof with total induction on $n$
for $n=1$ the statement is true by the definition. Suppose it is true for $n-1$ the for $n$
\begin{align}
&\mathbb{E}[\mu_{\Vy(j)}]=\sum_{\Vy \in \{1,\dots,d\}^n} \prod_{i=1}^n W(y_i|\alpha_j(i)) \mu_{\Vy}\\
&=\sum_{\Vy \in \{1,\dots,d\}^{n-1}} \prod_{i=1}^{n-1} W(y_i|\alpha_j(i)) \mu_{\Vy_1^{n-1}} \otimes \sum_{a \in \{1,\dots,d\}} W(a|\alpha_j(n))E_{a,a}
\end{align}
where $\Vy_1^{n-1}$ denote the first $n-1$ symbol of $\Vy$.
But by the definition the last quantity $\sum_{a \in \{1,\dots,d\}} W(a|\alpha_j(n))E_{a,a}$ is not else than $\rho_{\alpha_j(n)}$ so
\begin{align}
\mathbb{E}[\mu_{\Vy(j)}]=\sum_{\Vy \in \{1,\dots,d\}^{n-1}} \prod_{i=1}^{n-1} W(y_i|\alpha_j(i)) \mu_{\Vy_1^{n-1}} \otimes \rho_{\alpha_j(n)}
\end{align}
but for $n-1$ the statement is true, so
\begin{equation}\mathbb{E}[\mu_{\Vy(j)}]=\rho_{\alpha_j}\end{equation}
\end{proof}
\end{section}

\appendix
\begin{section}{proof of Lemma \ref{tip-proj}}
The proof based on typical sequences. These definition is a simplified/modified version of \cite{Csiszar}.
\begin{Def}[Typical sequence]
For a given probability distribution $P$ on $\{1,2,\dots,d\}$ an $\Vx\in \{1,2,\dots,d\}^n$ sequence is called \textit{P-typical} with constant $\delta$, if
\begin{equation}\left|\frac{1}{n}N(a|\Vx)-P(a)\right|\leq \frac{\delta}{\sqrt[4]{n}} \quad \textnormal{for every}\quad a \in \{1,2,\dots,d\} \end{equation}
where $N(a|\Vx)$ means the number occurrences of $a$ in sequence $\Vx$ and, in addition no $a \in \{1,2,\dots,d\}$ with $P(a)=0$ occurs. The set of such sequences will be denoted by $T^n_{[P]_\delta}$ or simply $T_{[P]}$.
\end{Def}
\begin{remark}\label{rem1}
If a sequence $\Vx$ is P-typical as above, then
\begin{equation}| -\sum_{a=1}^d \frac{N(a|\Vx)}{n}\log(P(a)) + \sum_{a=1}^d P(a)\log(P(a))|\leq \frac{Kd\delta}{\sqrt[4]{n}} \end{equation}
if $\delta$ is small enough, because if $P(b)$ is 0 then $N(b|\Vx)=0$ so $\frac{N(b|\Vx)}{n}-P(b)=0$ so the $b$-th element of the sum will be 0. If $P(a)>0$ then $\log(P(a))$ is finite, so $max_{a:P(a)>0}[-log(P(a))]=K$ is finite, so the above sum is smaller than $\frac{Kd\delta}{\sqrt[4]{n}}$.
\end{remark}
\begin{Lem}
For every distribution $P$ on $\{1,2,\dots,d\}$, and for every $\beta >0$
\begin{equation}P^n(T^n_{[P]_\delta})\geq 1-\beta \label{tip-h} \end{equation}
if $n$ is large enough.
\end{Lem}
\begin{proof}
If $\VX=X_1,X_2 ,\dots,X_n$ is an i.i.d. random sequence with distribution $P$ then the random variable $N(a|\VX)$ has the expectation value $nP(a)$ and variance $nP(a)(1-P(a))\leq \frac{n}{4}$.
Thus by the Chebishev's inequality
\begin{equation}\textnormal{Pr}\{\left|N(a|\VX)-n(Pa)|>n\frac{\delta}{\sqrt[4]{n}}\right|\leq \frac {1}{4\sqrt{n}\delta^2}\}\end{equation}
for every $a \in \{1,2,\dots,d\}$. From this the assertion follow.
\end{proof}
\begin{proof}[of Lemma \ref{tip-proj}]
with these typical sequences we can make typical subspace as follows:

Let $\sum_{i=1}^d \lambda_i |u_i\rangle \langle u_i|=\rho$ be a spectral decomposition of $\rho$.
Now define $P$ as $P(a)=\lambda_a a \in \{1,2,\dots,d\}$ (In this case $H(P)=S(\rho)$, and let $n$ be large enough to verify \ref{tip-h} with $\beta=\varepsilon/2$. Now we define $\Pi$ the typical projection of $\rho^{\otimes n}$ by as follows
\begin{equation}\Pi=\sum_{\Vx \in T^n_{[P]_\delta}} |u_{x_1}\rangle \langle u_{x_1}|\otimes |u_{x_2}\rangle \langle u_{x_2}|\otimes \dots \otimes |u_{x_n}\rangle \langle u_{x_n}|\end{equation} See that if a sequence $\Vx_1$ differs from $\Vx_2$ then the minimal projection generated by $\Vx_1$ is orthogonal to the minimal projection generated by $\Vx_2$.

For this projection $\hat\Pi$ the assertions (\ref{t/1}) and (\ref{hiv1}) of Lemma \ref{tip-proj} are valid.
See the first assertion
\begin{align}
\Tr(\hat\Pi \rho^{\otimes n} \hat\Pi)=&\Tr(\hat\Pi \rho^{\otimes n})=\\
=&\Tr(\sum_{\Vx \in T^n_{[P]_\delta}} |u_{x_1}\rangle \langle u_{x_1}| \rho \otimes |u_{x_2}\rangle \langle u_{x_2}|\rho \otimes \dots \otimes |u_{x_n}\rangle \langle u_{x_n}|\rho )=\\
=&\sum_{\Vx \in T^n_{[P]_\delta}} \prod_{i=1}^n \langle u_{x_i}| \rho |u_{x_i}\rangle =\sum_{\Vx \in T^n_{[P]_\delta}} \prod_{i=1}^n \lambda_{x_i}=\\
=&\sum_{\Vx \in T^n_{[P]_\delta}} \prod_{i=1}^n P(x_i)=P^n(T^n_{[P]_\delta}) \geq 1-\beta
\end{align}
if $n$ is large enough $n>N_1$, and the last row comes from the definition of $P$ and the previous Lemma.

Observe that (\ref{hiv1}) is true, because of Remark \ref{rem1}. See that spectrum of $\hat\Pi \rho \hat\Pi$ is equal with
\begin{equation}spect(\hat\Pi \rho \hat\Pi)=\{\prod_{a=1}^d \lambda_a^{N(a|\Vx)}, \Vx \in T^n_{[P]_\delta}\} \end{equation}
where
\begin{equation}\prod_{a=1}^d \lambda_a^{N(a|\Vx)}=2^{-n\sum_{a=1}^d -\frac{N(a|\Vx)}{n} \log(\lambda_a) } \end{equation}
and from Remark \ref{rem1} we know that for all $\Vx \in T^n_{[P]_\delta}$
\begin{equation}\prod_{a=1}^d \lambda_a^{N(a|\Vx)}\leq 2^{-n(S(\rho)-\frac{Kd\delta}{\sqrt[4]{n}})  }\end{equation}
where $\frac{Kd\delta}{\sqrt[4]{n}}$ is smaller than $\beta$ if $n$ large enough.

We know from \cite{Petz} that, if $n$ is large enough, there is an another projection $\tilde\Pi$ which satisfy (\ref{t/1}), (\ref{t/3}). Now the projection which satisfy all the assertion of the Lemma is given by $\Pi=\hat\Pi\wedge \tilde\Pi$, because
\begin{align}\Tr(\Pi^c \rho \Pi^c)\leq& \Tr(\hat\Pi^c \rho \hat\Pi^c)+\Tr(\tilde\Pi^c \rho \tilde\Pi^c)=2\beta\\
	\|\Pi \rho \Pi \| \leq& \|\hat\Pi \rho \hat\Pi \| \leq 2^{-n(\S(\rho)-\beta)}\\
	\Tr(\Pi \sigma^{\otimes n} \Pi) \leq& \Tr(\tilde\Pi \sigma^{\otimes n} \tilde\Pi) \leq 2^{-n(D(\rho\|\sigma)-\beta)}
\end{align}

\end{proof}
\end{section}

\end{document}